\documentclass[showpacs,amsmath,amssymb,twocolumn,pra,superscriptaddress]{revtex4-1}

\usepackage{amsfonts}
\usepackage{amssymb}
\usepackage{amscd}
\usepackage{amsmath}    
\usepackage{multirow}
\usepackage{graphicx}
\usepackage{dcolumn}
\usepackage{bm}
\usepackage{amsthm}
\usepackage{fullpage}
\usepackage{microtype}
\usepackage[libertine]{newtxmath}
\usepackage[tt=false]{libertine} 
\usepackage{caption}
\usepackage{bbm}
\usepackage{hyperref, color}
\hypersetup{colorlinks=true,citecolor=blue, linkcolor=blue, urlcolor=blue}
\usepackage[linesnumbered,boxed,ruled,vlined]{algorithm2e}
\usepackage{bbm}
\usepackage{xcolor}
\usepackage{enumerate}
\usepackage{enumitem}
\usepackage{tabularx}
\usepackage{array}
\usepackage{cleveref}
\newcolumntype{L}[1]{>{\raggedright\arraybackslash}p{#1}}
\newcolumntype{C}[1]{>{\centering\arraybackslash}m{#1}}
\newcolumntype{R}[1]{>{\raggedleft\arraybackslash}p{#1}}

\usepackage{makecell}

\usepackage{cleveref}
\usepackage{color}
\usepackage{graphicx}
\usepackage{float}
\usepackage{bbm}

\usepackage{footnote}

\newtheorem{theorem}{Theorem}[section]

\newtheorem{corollary}[theorem]{Corollary}
\newtheorem{lemma}[theorem]{Lemma}
\newtheorem{definition}[theorem]{Definition}

\renewcommand{\vec}[1]{\mathbf{#1}}
\renewcommand{\Pr}{P}

\newcommand{\var}{\mathrm{Var}}

\newcommand{\Ext}{\mathrm{Ext}}
\newcommand{\sm}{\eta}
\newcommand{\rbs}{\mathrm{rbs}}
\newcommand{\rmmneq}{\mathrm{neq}}
\newcommand{\rmmeq}{\mathrm{eq}}
\renewcommand{\ggg}{\mathrm{g}}

\renewcommand{\vec}[1]{\boldsymbol{#1}}
\newcommand{\qstate}{\mathcal{E}}
\newcommand{\gqstate}{\mathcal{E}}

\newcommand{\bra}[1]{\mbox{$\left\langle #1 \right|$}}
\newcommand{\ket}[1]{\mbox{$\left| #1 \right\rangle$}}

\begin{document}
\preprint{APS/123-QED}
\title{Improved Real-Time Post-Processing for Quantum Random Number Generators}

\begin{abstract}
Randomness extraction is a key problem in cryptography and theoretical computer science.
With the recent rapid development of quantum cryptography, quantum-proof randomness extraction has also been widely studied, addressing the security issues in the presence of a quantum adversary.
In contrast with conventional quantum-proof randomness extractors characterizing the input raw data as min-entropy sources, we find that the input raw data generated by a large class of trusted-device quantum random number generators can be characterized as the so-called reverse block source. This fact enables us to design improved extractors. Specifically,
we propose two novel quantum-proof randomness extractors for reverse block sources that realize real-time block-wise extraction. In comparison with the general min-entropy randomness extractors, our designs achieve a significantly higher extraction speed and a longer output data length with the same seed length. In addition, they enjoy the property of online algorithms, which process the raw data on the fly without waiting for the entire input raw data to be available. 
These features make our design an adequate choice for the real-time post-processing of practical quantum random number generators. Applying our extractors to the raw data generated by a widely used quantum random number generator, we achieve a simulated extraction speed as high as $300$ Gbps.
\end{abstract}

\author{Qian Li}
\affiliation{Shenzhen lnternational Center For Industrial And Applied  Mathematics, Shenzhen Research Institute of Big Data, Shenzhen, China.}

\author{Xiaoming Sun}
\affiliation{State Key Lab of Processors, Institute of Computing Technology, Chinese Academy of Sciences, 100190, Beijing, China.}

\author{Xingjian Zhang}
\affiliation{Center for Quantum Information, Institute for Interdisciplinary Information Sciences,
Tsinghua University, Beijing 100084, China}

\author{Hongyi Zhou}
\email{zhouhongyi@ict.ac.cn}
\affiliation{State Key Lab of Processors, Institute of Computing Technology, Chinese Academy of Sciences, 100190, Beijing, China.}

\maketitle

\section{introduction}
Randomness extraction aims at distilling uniform randomness from a weak random source \cite{impagliazzo1989recycle}, which is widely applied ranging from cryptography to distributed algorithms.
Recently, quantum cryptography \cite{bennett1984quantum} has been developed rapidly, whose security is guaranteed by the fundamental principle of quantum mechanics \cite{lo1999Unconditional,shor2000Simple}. Compared with its classical counterpart, the unique characteristic of intrinsic randomness enables quantum cryptography with the feasibility of secure communication regardless of the computation power of the eavesdroppers. Randomness extraction also serves as a key step, privacy amplification \cite{bennett1995generalized,bennett1988privacy}, in the post-processing of quantum cryptography, eliminating the side information possessed by a quantum adversary.

The input raw data is conventionally characterized as a min-entropy source. Randomness extraction from a min-entropy source is realized by the so-called min-entropy extractors.
In quantum cryptography, only quantum-proof extractors can provide information-theoretic security in the presence of a quantum adversary \cite{konig2011sampling}. Among all the quantum-proof min-entropy extractors, Trevisan's extractor \cite{trevisan2001extractors,de2012trevisan} and the Toeplitz-hashing extractor \cite{mansour1993computational} are two of the most popular choices. Both of them are strong extractors \cite{nisan1996randomness}, which means that the extracted randomness is independent of the seed. 
Trevisan's extractor requires a seed length of only polylogarithmic scaling of the input length, which is much lower than the linear scaling for the Toeplitz-hashing extractor. On the other hand, the output speed is much lower than that of the Toeplitz-hashing extractor \cite{ma2013postprocessing}. The Toeplitz-hashing extractor is constructed from a cyclic matrix, which is easy to implement and can be accelerated by the Fast Fourier Transformation (FFT) \cite{nussbaumer1981fast,golub1996matrix}.

In order to save the seed and accelerate the extraction speed, a naive post-processing method is to apply the quantum-proof min-entropy extractors (e.g., the Toeplitz-hashing extractor) in a block-wise way: divide the raw data into multiple small blocks, and then repeatedly use a short seed for each block. However, we should point out that this kind of block-wise post-processing on a general min-entropy source could be insecure, i.e., the output could be far from a uniform distribution. For example, consider a $N$-bit input raw data  where the first $N/2$ bits are uniformly random and the last $N/2$ bits are a replication of the first $N/2$ bits (in the same order). Then it is straightforward that the min-entropy of the $N$-bit raw data is $N/2$. However, suppose we divide the raw data into two blocks each of $N/2$ bits, and apply the Toeplitz-hashing extractor to each block with the same seed, then the output of the last block is a replication of that of the first block. This means the whole output is far from a uniform distribution and thus insecure.

For a general min-entropy source, the correct way  of extraction is to apply a quantum-proof min-entropy extractor once on the whole input raw data. However, in the practical implementation, there are three subtle issues. The first two issues are about the computational efficiency. First, both Toeplitz-hashing extractors and Trevisan's extractors are not online: they cannot output even one bit until all bits of the raw data are available. Second, the current extraction speed would be much lower than the raw data generation speed in most quantum random number generator (QRNG) implementations \cite{3GQRNG}. The generation speed of the extracted quantum random numbers is restricted by the extraction speed in the post-processing, which becomes a bottleneck in the applications of quantum random numbers in quantum communication tasks. The third issue is about the seed: the seed used in the extractor is a limited resource, and the output data length is restricted by the seed length, especially for some commercial QRNGs where the seed cannot be updated. We want to extract as many random numbers as possible with a limited seed length.

To deal with the practical issues above, we consider properties of some specific randomness sources beyond the conventional min-entropy source characterization.
In this work, we design two novel quantum-proof randomness extraction algorithms for a large class of QRNGs where the raw data can be described as the reverse block source \cite{Vadhan03}. Our results are inspired by the extractors designed in \cite{CG88,survey} for block sources and Santha-Vazirani sources, respectively. Both of our two extractors are \emph{online} algorithms: as the input raw data arrives in real-time, the two extractors proceed on-the-fly input raw data piece-by-piece, where the processing is independent of the data in the future. In fact, our two extractors are block-wise extractors. That is, they partition the input raw data into blocks serially in the time order the input arrives, and then apply a min-entropy extractor to each block.
Moreover, suppose the input raw data is of length $N$, then our first extractor requires a seed length of $O(\log N)$ and takes equipartition of the block lengths, and our second extractor requires only a seed length of $O(1)$ and incremental block lengths. Compared to the first extractor, the second one scarifies the extraction speed while enjoying a seed length independent of the input length. As a result, this extractor can deal with infinite raw data, without the need of updating the seed and determining the raw data length prior to extraction. For both extractors, the output length is a constant fraction of the min-entropy of the raw data, which indicates that our extractors are quite efficient. To show the performances of the extractors, we make a simulation estimating the output speed. It turns out that the extraction speed is adequate for the post-processing of the common implementations of QRNGs .

\section{Preliminaries}
Throughout the paper, we use capital and lowercase letters to represent random variables and their assignments, respectively. 
We use $U_m$ to represent the perfectly uniform random variable on $m$-bit strings and $\rho_{U_m}$ to represent the $m$-dimensional maximally mixed state.

\begin{definition}[Conditional min-entropy]
Let $Y$ be a classical random variable that takes value $y$ with probability $\Pr_y$ and $\gqstate$ be a quantum system. 
The state of the composite system can be written as $\rho_{Y\gqstate}=\sum_y \Pr_y\ket{y}\bra{y}\otimes \rho_{\gqstate}^y$, where $\{\ket{y}\}_y$ forms an orthonormal basis. 
The conditional min-entropy of $Y$ given $\gqstate$ is
$H_{\min}(Y|\gqstate)_{\rho_{Y\gqstate}}=-\log_2 p_{\mathrm{guess}}(Y|\gqstate)$, where $p_{\mathrm{guess}}(Y|\gqstate)$ is the maximum average probability of guessing $Y$ given the quantum system $\gqstate$. That is,
\begin{equation}
p_{\mathrm{guess}} (Y|\gqstate)=\max_{\{E_{\gqstate}^y\}_y}\left[\sum_y\Pr_y\mathrm{Tr}\left(E_{\gqstate}^y\rho_{\gqstate}^y\right)\right],
\end{equation}
where the maximization is taken over all positive operator-valued measures (POVMs) $\{E_{\gqstate}^y\}_y$ on $\gqstate$.
\end{definition}

When system $Y$ is decoupled from $\gqstate$, where $\rho_{Y\gqstate}=\sum_y \Pr_y\ket{y}\bra{y}\otimes \rho_{\gqstate}$, the conditional min-entropy of $Y$ given $\gqstate$ reduces to the classical min-entropy, $H_{\min}(Y)=-\log_2 \max_y \Pr_y$. When $\rho_{Y\gqstate}$ is clear from the context, we will denote the conditional min-entropy as $H_{\min}(Y|\gqstate)$ for brevity.




In this paper, we call the raw data generated by a QRNG as a random source. A general random source is the min-entropy source, where the conditional min-entropy is lower bounded.

\begin{definition}[Min-entropy quantum-proof extractor]
A function $\Ext:\{0,1\}^{n}\times \{0,1\}^d\rightarrow \{0,1\}^m$ is a $(\delta n,\epsilon)$ min-entropy quantum-proof extractor, if for every random source $Y$ and quantum system $\gqstate$ satisfying $H_{\min}(Y|\gqstate)\geq \delta n$, we have
\begin{equation}\label{eq:criteria_ext}
\frac{1}{2}\|\rho_{\Ext(Y,S)\gqstate} - \rho_{U_m}\otimes \rho_{\gqstate} \| \leq \epsilon,
\end{equation}
where $S$ is called the seed, which is a perfectly uniform random variable on $d$-bit strings independent of the system $Y\gqstate$ and $\|\cdot\|$ denotes the trace norm defined by $\|A\| = \mathrm{Tr}\sqrt{A^\dag A}$. An extractor $\Ext$ is said to be strong if
\begin{equation}\label{eq:criteria_strongext}
\frac{1}{2}\|\rho_{\Ext(Y,S)S\gqstate} - \rho_{U_m}\otimes \rho_{U_d}\otimes \rho_{\gqstate} \| \leq \epsilon.
\end{equation}
We call the concatenation of the output string of a strong extractor with the seed as an expansion, denoted as the tuple
$\left(\Ext(y,s), s\right)$.
\end{definition}

It is straightforward to see that an expansion is a standard $(\delta n,\epsilon)$ min-entropy quantum-proof extractor.
If the output of an extractor satisfies Eq.~\eqref{eq:criteria_ext} or \eqref{eq:criteria_strongext}, we say that the output is $\epsilon$-close to a uniform distribution.

A widely used randomness extractor is the Toeplitz-hashing extractor.

\begin{definition}[Toeplitz-hashing extractor]
A $u\times n$ matrix $T$ is a Toeplitz matrix if $T^{ij}=T^{i+1,j+1}=s_{j-i}$ for all
$i=1,\cdots,u-1$ and $j=1,\cdots,n-1$. A Toeplitz matrix over the finite field $GF(2)$, $T_{s}$, can be specified by a bit string $s=(s_{1-u},s_{2-u},\cdots,s_{n-1})$ of length $u+n-1$.
Given any $n,d\in \mathbb{N}^+$ where $d\geq n$, define the Toeplitz-hashing extractor $\Ext_{T}^{n,d}:\{0,1\}^n\times \{0,1\}^d\rightarrow \{0,1\}^{d-n+1}$ as $\Ext_{T}^{n,d}(y,s)=T_s\cdot y$, and define the expanded Toeplitz-hashing extractor $\Ext_{T'}^{n,d}:\{0,1\}^n\times \{0,1\}^{d}\rightarrow \{0,1\}^{2d-n+1}$ as $\Ext_{T'}^{n,d}(y,s)=(T_{s}\cdot y,s)$, where the matrix product operation $\cdot$ is calculated over the field $GF(2)$.
\end{definition}

Since $\{T_{s}\cdot y|s\in\{0,1\}^{u+n-1}\}$ is a family of pairwise independent hashing functions \cite{MansourNT90,Krawczyk94}, according to the quantum Leftover Hash Lemma \cite{TomamichelRSS10}, we can prove that the Toeplitz-hashing extractor is a min-entropy quantum-proof strong extractor.


\begin{lemma}[\cite{TomamichelRSS10}]\label{coro:gadget}
For every $n\in\mathbb{N}^+$ and $\delta>0$, $\Ext_{T}^{n,d}$ is a $\left(\delta n,\epsilon\right)$ min-entropy quantum-proof strong extractor, where $\epsilon=2^{-(\delta n+n-d-1)/2}$. Equivalently, $\Ext_{T'}^{n,d}$ is a $\left(\delta n,\epsilon\right)$ min-entropy quantum-proof extractor.
\end{lemma}

Note that the output of $\Ext_{T'}^{n,d}$ has $2d-n+1$ bits, where the last $d$ bits form the seed. We remark that though the Toeplitz-hashing extractor $\Ext_T^n$ is strong, the expanded Toeplitz-hashing extractor $\Ext_{T'}^n$ is not.

\section{Main result}

We use $X=X_1X_2\cdots X_N\in (\{0,1\}^b)^N$ to denote the raw data generated by a QRNG, which contains $N$ samples each of $b$ bits. For a set $I\subset \mathbb{N}^+$, we write $X_I$ for the restriction
of $X$ to the samples determined by $I$. For example, if $I=\{2,3,5\}$, then $X_I=X_2X_3X_5$.
We use $\qstate$ to denote the quantum system possessed by the quantum adversary.

\subsection{Reverse block source}
For a given randomness extractor, there is a trade-off between its performance and the generality of the random sources it applies to. The more special random sources the extractor works for, the better performance the extractor may achieve. In this paper, the notion of \emph{reverse block sources}, which are more special than the min-entropy sources, plays a critical role in the sense that (i) raw data of a large class of QRNGs can be described as a reverse block source and (ii) pretty good quantum-proof extractors for reverse block sources exist. A quantum version of the reverse block source is defined below.
\begin{definition}[Reverse block source, adapted from Definition 1 in \cite{Vadhan03}]\label{def:rbs}
	A string of random variables $X=X_1\cdots X_N\in(\{0,1\}^b)^N$ is a $(b,N,\delta)$-reverse block source given a quantum system $\qstate$ if 
for every $1\leq k\leq i\leq N$ and every $x_{i+1},x_{i+2},\cdots, x_N$,
\begin{equation}\label{eq:rrrbs_def}
\begin{aligned}
&H_{\min}(X_k,X_{k+1},\cdots,X_{i}| X_{i+1}=x_{i+1},\cdots,X_N=x_N,\qstate) \\
&\geq (i-k+1)\delta\cdot b.
\end{aligned}
\end{equation}
\end{definition}
Intuitively, Eq.~\eqref{eq:rrrbs_def} implies 
there are at least $\delta b$ bits of information in $X_i$ that will be “forgotten” at the next time step. Consider a (forward) block source satisfying 
$H_{\min}(X_{k},\cdots,X_i|X_{k-1}=x_{k-1},\cdots,X_1=x_1,\qstate)\geq (k-i+1)\delta b.
$
A reverse block source can be understood as a time-reversed (forward) block source. A reverse block source is somewhat less natural than the standard notion of the (forward) block sources, but as we will see later, it may still be a reasonable model for some physical sources. In fact, the consideration of the reverse-block condition raises interesting philosophical questions: does the universe keep a perfect record of the past? If not, then the reverse-block condition seems plausible \cite{Vadhan03}. 

For QRNGs where the raw data are mutually independent, such as the ones based on single photon detection \cite{Rarity94, Stefanov00, Jennewein00}, vacuum fluctuations \cite{Gabriel10,Symul11,Jofre11}, and photon arrival time \cite{Wahl11,Li13,Nie14}, the min-entropy source automatically satisfies Eq.~\eqref{eq:rrrbs_def}, where $\forall x_{i+1},\cdots,x_N$, 
$H_{\min}(X_{k},\cdots,X_i| X_{i+1}=x_{i+1},\cdots,X_N=x_N,\qstate)=\sum_{j=k}^i H_{\min}(X_j| \qstate)$, 
hence is also a reverse block source. For QRNGs with correlated raw data, one can construct appropriate physical models to check whether Eq.~\eqref{eq:rrrbs_def} is satisfied. In Appendix~\ref{app:rbs}, we take a common QRNG scheme employing homodyne detectors with finite bandwidth~\cite{gehring2021homodyne}, as an example to prove Eq.~\eqref{eq:rrrbs_def}. More generally, for physical systems with time-reversal symmetry, the reverse-block condition holds if and only if the forward-block condition holds. In Sec.~\ref{sec:ext_rbs}, we will see why the reverse-block condition (high entropy conditioned on the future) is necessary for designing online algorithms.


We also consider a smoothed version of the reverse block source.
For $X=X_1\cdots X_N\in(\{0,1\}^b)^N$, Denote the underlying joint quantum state of the random source over the systems $X$ and $\qstate$ as $\rho_{X\qstate}$. We call $X$ a $(b,N,\delta,\epsilon_\mathrm{s})$-smoothed reverse block source, 
if there exists a state $\rho_{X\qstate}^*$ that is $\varepsilon_s$-close to $\rho_{X\qstate}$, 
\begin{equation}
    \frac{1}{2}\|\rho_{X\qstate}^*-\rho_{X\qstate}\|\leq\varepsilon_s,
\end{equation}
such that $\rho_{X\qstate}^*$ is a $(b,N,\delta)$-reverse block source.
The smoothed reverse block source is in the same spirit of the smooth conditional entropy~\cite{renner2008security}. The motivation of introducing a smoothed version is to exclude singular points or regions in a probability distribution, which will help extract more randomness.

Here we remark that we mainly consider the trusted-device QRNGs where the extraction speed is a bottleneck. For QRNGs with a higher security level, such as the semi-device-independent QRNGs and device-independent QRNGs, the reverse block source property is not satisfied in general. These types of QRNGs require fewer assumptions and characterizations on the devices at the expense of relatively low randomness generation rates of raw data. Then their real-time randomness generation rates are not limited by the extraction speed.
\subsection{Extractors for reverse block sources}\label{sec:ext_rbs}

In this section, we design two online quantum-proof extractors that can both extract a constant fraction of the min-entropy from reverse block sources. The two extractors both proceed in the following fashion: partition the input raw data into blocks, and apply a min-entropy quantum-proof extractor to each block using part of the output of the previous block as the seed. The basic building block of our designs is family of $(\delta n,\epsilon_n)$ min-entropy quantum-proof extractors, denoted by $\Ext_\ggg=\{\Ext_\ggg^n:n\in\mathbb{N}^+\}$. Here, the subscript of $\Ext_{\ggg}$ means ``gadget". Let $d_n$ and $m_n$ denote the seed length and output length of $\Ext_\ggg^n$, respectively, i.e., $\Ext_\ggg^n:\{0,1\}^n\times\{0,1\}^{d_n}\rightarrow \{0,1\}^{m_n}$. We require the gadget $\Ext_\ggg^n$ to satisfy: (i) $\epsilon_n$ is exponentially small, for instance, $\epsilon_n=2^{-\delta n/4}$, which aims at that the summation $\sum_{n=1}^{\infty} \epsilon_n$ converges; and (ii) $m_n-d_n\geq c\delta n$ for some constant $c>0$, which means that $\Ext_\ggg^n$ extracts a constant fraction of min-entropy from the raw data. As an explicit construction, $\Ext_\ggg^n$ can be specified as the expanded Toeplitz-hashing extractor $\Ext_{T'}^{n,d_n}$ where $d_{n}=(1+\delta/2)n-1$, then $m_n-d_n=\delta n/2$, $m_n=(1+\delta)n-1$, and $\epsilon_n=2^{-\delta n/4}$. In the rest of this paper, we abbreviate $\Ext_{T}^{n,d_n}$ and $\Ext_{T'}^{n,d_n}$ where $d_{n}=(1+\delta/2)n-1$ to $\Ext_T^n$ and $\Ext_{T'}^n$, respectively. Besides, to simplify our discussion, we assume $\epsilon_n=2^{-\delta n/4}$, and the analysis for other exponentially small values of $\epsilon_n$ is similar.

The two extractors are described in Algorithms \ref{alg:rbseq} and \ref{alg:rbsneq}, respectively. Given an input raw data of length $N$, the first extractor, named $\Ext_{\rbs}^{\rmmeq}$, evenly partitions the input raw data into blocks each of size $O(\log N)$ and requires $O(\log N)$ random bits as the initial seed. In particular, if the min-entropy quantum-proof extractor $\Ext_\ggg$ in use is the expansion of a strong extractor such as the expanded Toeplitz-hashing extractor, then $\Ext_{\rbs}^{\rmmeq}$ degenerates exactly to the following naive extractor: partition the input raw data into equal-sized blocks and apply the corresponding strong extractor to each block separately with the same seed. The second extractor, named $\Ext_{\rbs}^{\rmmneq}$, is inspired by the extractor that can extract randomness from Santha-Vazirani sources using a seed of constant length \cite{survey}. It uses only $O(1)$ random bits as the initial seed and requires incremental block lengths. Compared to the first extractor, the second one is less hardware-friendly and sacrifices the extraction speed in general. On the other hand, it enjoys the property that the seed length is independent of the input length. As a result, this extractor can deal with infinite raw data, without the need of determining the raw data length prior to extraction. In other words, one does not need to update the seed in practical implementations. We remark that the initial seed is indispensable because there does not exist any nontrivial deterministic extractor for reverse block sources. The proof is presented in Appendix \ref{app:seed_is_necessary}. 

\begin{algorithm}[htb]
	\caption{$\Ext_{\rbs}^{\rmmeq}$}
	\label{alg:rbseq}
	\textbf{Input}: $b\in\mathbb{N}^+$ and $0<\epsilon,\delta<1$. A string $x=x_1,x_2,\cdots,x_N$ sampled from a $(b,N,\delta)$-reverse block source;\\
	Let $i:=1$ and $n:=\left\lceil\frac{4}{\delta b}\cdot \log\left(\frac{N}{\epsilon}\right)\right\rceil$;\\
	Sample a uniform random bit string $s^{(1)}$ of length $d_{bn}$;\\
	\For{$\ell=1$ to $N/n$}{
		Let $I_{\ell}:=[i,i+n-1]$;\\
		Compute $z^{(\ell)}:=\Ext_\ggg^{bn}(x_{I_{\ell}},s^{(\ell)})$;\\
		Let $i:=i+n$;\\
		Cut $z^{(\ell)}$ into two substrings, denoted by $r^{(\ell)}$ and $s^{(\ell+1)}$, of size $m_{bn}-d_{bn}$ and $d_{bn}$ respectively;\\
		Output $r^{(\ell)}$.
	}
\end{algorithm}

\begin{algorithm}[htb]
	\caption{$\Ext_{\rbs}^{\rmmneq}$}
	\label{alg:rbsneq}
	\textbf{Input}: $b\in\mathbb{N}^+$ and $0<\delta<1$. A string $x=x_1,x_2,\cdots$ sampled from a $(b,\infty,\delta)$-reverse block source;\\
	\textbf{Parameter:} $n_1,\Delta\in\mathbb{N}^+$; \\
	Let $i:=1$;\\
	Sample a uniform random bit string $s^{(1)}$ of length $d_{bn_1}$;\\
	\For{$\ell=1$ to $\infty$}{
		Let $I_{\ell}:=[i,i+n_{\ell}-1]$;\\
		Compute $z^{(\ell)}:=\Ext_\ggg^{bn_{\ell}}(x_{I_{\ell}},s^{(\ell)})$;\\
		Let $i:=i+n_{\ell}$ and $n_{\ell+1}:=n_{\ell}+\Delta$;\\
		Cut $z^{(\ell)}$ into two substrings, denoted by $r^{(\ell)}$ and $s^{(\ell+1)}$, of size $m_{bn_\ell}-d_{bn_{\ell+1}}$ and $d_{bn_{\ell+1}}$ respectively\footnote{We should choose the parameters $n_1$ and $\Delta$ properly to have $m_{bn_\ell}-d_{bn_{\ell+1}}\geq1$.};\\
		Output $r^{(\ell)}$.
	}
\end{algorithm}

\begin{figure*}[htbp]
	\centering
	\includegraphics[scale=0.65]{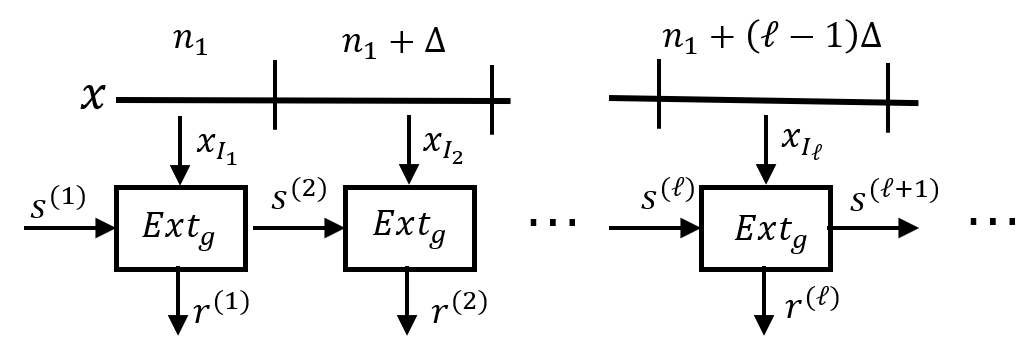}
	\caption{Illustration of $\Ext_{\rbs}^{\rmmneq}$.}
	\label{fig3}
\end{figure*}

Note that if we impose $\Delta=0$ and let $n_1=\left\lceil\frac{4}{\delta b}\cdot \log\left(\frac{N}{\epsilon}\right)\right\rceil$, then $\Ext_{\rbs}^{\rmmneq}$ becomes $\Ext_{\rbs}^{\rmmeq}$.
We first analyze the second extractor $\Ext_{\rbs}^{\rmmneq}$.

\begin{theorem}\label{thm:ext_rbs}
The extractor $\Ext_{\rbs}^{\rmmneq}$ satisfies the following properties:
\begin{itemize}
\item[(a)] It uses a seed of $d_{bn_1}$ length.
\item[(b)] For any $k\in\mathbb{N}$, we have
\begin{equation}		\frac{1}{2}\|\rho_{r^{(1)}\circ r^{(2)}\circ\cdots\circ r^{(k)}\qstate}-\rho_{U_{\sm_k}}\otimes \rho_{\qstate}\|\leq \sum_{\ell=1}^k2^{-\delta bn_\ell/4}<\frac{2^{-\delta b n_1/4}}{1-2^{-\delta b\Delta/4}},
\end{equation}
where $\sm_k=\sum_{\ell=1}^{k}(m_{bn_{\ell}}-d_{bn_{\ell+1}})$.
\end{itemize}
\end{theorem}

\begin{proof} The nontrivial part is Part (b). For convenience of presentation, we use $I_{k:\infty}$ to represent $I_k\cup I_{k+1}\cup\cdots\cup I_{\infty}$. In fact, we will prove that for any $k\in\mathbb{N}$ it has
\begin{equation}\label{eq:induction} 
\begin{aligned}
&\frac{1}{2}\left\|\rho_{r^{(1)}\circ r^{(2)}\circ\cdots\circ r^{(k)}\circ s^{(k+1)}\circ X_{I_{k+1:\infty}}\qstate}-\rho_{U_{\sm_k+d_{bn_{k+1}}}}\otimes \rho_{X_{I_{k+1:\infty}}\qstate}\right\| \\
&\leq \sum_{\ell=1}^k 2^{-\delta bn_\ell/4},
\end{aligned}
\end{equation}
which implies Part (b) immediately.
	
	The proof is by an induction on $k$. The base case when $k=0$ is trivial. The induction proceeds as follows. Suppose Eq.~\eqref{eq:induction} is true.
	Then due to the contractivity of trace-preserving quantum operations, we have
	\begin{equation}\label{eq:proof1}
    \begin{aligned}
	&\frac{1}{2}\left\|\rho_{r^{(1)}\circ r^{(2)}\circ\cdots\circ r^{(k)}\circ \Ext_\ggg\left(X_{I_{k+1}},s^{(k+1)}\right)\circ X_{I_{k+2:\infty}}\qstate} \right.\\
      &\quad \left.-\rho_{U_{\sm_k}}\otimes \rho_{\Ext_\ggg\left(X_{I_{k+1}},U_{d_{bn_{k+1}}}\right)\circ X_{I_{k+2:\infty}}\qstate}\right\| \\
      &\leq \sum_{\ell=1}^k 2^{-\delta bn_\ell/4}.
 \end{aligned}
	\end{equation}
	On the other hand, by Definition \ref{def:rbs}, for any assignment $x_{I_{k+2:\infty}}$ of $X_{I_{k+2:\infty}}$, we have
	\begin{equation}
	H_{\min}(X_{I_{k+1}}| X_{I_{k+2:\infty}}=x_{I_{k+2:\infty}},\qstate)\geq \delta bn_{k+1}.
	\end{equation}
	Then, recalling that $\Ext_\ggg$ is a min-entropy quantum-proof extractor, it follows that
\begin{equation}
\begin{aligned}	
&\frac{1}{2}\left\|\rho_{\Ext_\ggg\left(X_{I_{k+1}},U_{d_{bn_{k+1}}}\right)\circ X_{I_{k+2:\infty}}\qstate}-\rho_{U_{m_{bn_{k+1}}}}\otimes \rho_{X_{I_{k+2:\infty}}\qstate}\right\| \\
&\leq 2^{-\delta bn_{k+1}/4}.
\end{aligned}
\end{equation}
Thus,
\begin{equation}
	\begin{aligned}\label{eq:proof2} 
 &\frac{1}{2}\left\|\rho_{U_{\sm_k}}\otimes\rho_{\Ext_\ggg\left(X_{I_{k+1}},U_{d_{bn_{k+1}}}\right)\circ X_{I_{k+2:\infty}}\qstate} \right. \\
 &\quad\left.-\rho_{U_{\sm_k}}\otimes\rho_{U_{m_{bn_{k+1}}}}\otimes \rho_{X_{I_{k+2:\infty}}\qstate}\right\| \\
 & \leq 2^{-\delta bn_{k+1}/4}.
\end{aligned}
\end{equation}
Finally, combining inequalities \eqref{eq:proof1} and \eqref{eq:proof2} and applying the triangle inequality, we conclude that
\begin{equation} 
\begin{aligned}
&\frac{1}{2}\left\|\rho_{r^{(1)}\circ r^{(2)}\circ\cdots\circ r^{(k+1)}\circ s^{(k+2)}\circ X_{I_{k+2:\infty}}\qstate} \right. \\
& \left. \quad -\rho_{U_{\sm_{k+1}+d_{bn_{k+2}}}}\otimes \rho_{X_{I_{k+2:\infty}}\qstate}\right\| \\
&\leq \sum_{\ell=1}^{k+1} 2^{-\delta bn_\ell/4}.
\end{aligned}
\end{equation}
By using the summation formula for the geometric progression, the above inequality is further upper bounded by $\frac{2^{-\delta b n_1/4}}{1-2^{-\delta b\Delta/4}}$.
\end{proof}

As can be seen from the proof, the parameter $\Delta$ of $\Ext_{\rbs}^{\rmmneq}$ must be strictly positive, since the upper bound $\sum_{\ell=1}^{k} 2^{-\delta bn_\ell/4}$ on the error converges if and only if $\Delta>0$. Theorem \ref{thm:ext_rbs} implies that the (infinitely long) output string $r^{(1)}\circ r^{(2)}\circ\cdots$ can be arbitrarily close to the uniform distribution by choosing a sufficiently large constant $n_1$. As an explicit construction, suppose $\Ext_{\rbs}^{\rmmneq}$ adopts the expanded Toeplitz-hashing extractor $\Ext^n_{T'}$ as the gadget $\Ext^n_\ggg$, where $d_{n}=(1+\delta/2)n-1$ and $m_n=(1+\delta)n-1$. We further set $\Delta=1$. We require that $n_1\geq 4/\delta+1$ such that $m_{bn_\ell}-d_{bn_{\ell+1}}=b(\delta n_{\ell}/2-1-\delta/2)\geq 1$ for any $\ell$. Then $\Ext_{\rbs}^{\rmmneq}$ extracts $(1+\delta)b(n_1+\ell-1)-1$ random bits from the $\ell$-th block $x_{I_{\ell}}$, outputs the first $m_{bn_\ell}-d_{bn_{\ell+1}}\approx \delta b n_{\ell}/2$ bits, and then uses the last $(1+\delta/2)b(n_1+\ell)-1$ bits as the seed of the next block.

We intuitively explain Theorem~\ref{thm:ext_rbs}, i.e., how our algorithms work for reverse block sources. When the input raw data is a reverse block source, the block $X_{I_{1}}$ has non-zero min-entropy conditioned on any assignment of future blocks, which implies that the output $z^{(1)}$ of this block is (approximately) uniform conditioned on any assignment of future blocks and further implies that $z^{(1)}$ is independent of future blocks. Then $r^{(1)}$, $s^{(2)}$, and the further blocks are mutually independent, which guarantees the following two properties:
 (i) the seed $s^{(2)}$ and the raw data $X_{I_{2}}$ are independent when applying $\Ext_\ggg$ to $X_{I_{2}}$; 
(ii) the output $r^{(1)}$ and $r^{(2)}$ are independent, as $r^{(2)}$ is fully determined by $s^{(2)}$ and $X_{I_{2}}$.
 By repeating the above argument to $X_{I_2},X_{I_3},\cdots$ and so on, we can see that each $r^{(k)}$ is uniform and $r^{(1)},\cdots,r^{(k)},\cdots$ are mutually independent. So we can conclude that the whole output $r^{(1)}\circ\cdots\circ r^{(k)}\circ\cdots$ is a uniform distribution. We remark that the reverse-block condition is necessary, as it guarantees that the seed and the raw data are independent when applying $\Ext_\ggg$ while the (forward) block condition does not.


Via a similar argument as in Theorem \ref{thm:ext_rbs}, we have the following result for $\Ext_{\rbs}^{\rmmeq}$.
\begin{theorem}\label{thm:ext_rbseq}
The extractor $\Ext_{\rbs}^{\rmmeq}$ uses $d_{bn}$ random bits as a seed and outputs a string $r^{(1)}\circ r^{(2)}\circ\cdots\circ r^{(N/n)}$ satisfying that
\begin{equation}
  \frac{1}{2}\|\rho_{r^{(1)}\circ r^{(2)}\circ\cdots\circ r^{(N/n)}\circ\qstate}-\rho_{U_{\eta}}\otimes \rho_{\qstate}\|\leq \frac{N}{n}\cdot 2^{-\delta bn/4}\leq \epsilon,
\end{equation}
where $\eta:=\frac{N}{n}\cdot (m_{bn}-d_{bn})$. 
\end{theorem}
In particular, suppose $\Ext_{\rbs}^{\rmmeq}$ adopts $\Ext^n_{T'}$ as the gadget $\Ext_\ggg^n$. Then $\Ext_{\rbs}^{\rmmeq}$ extracts $(1+\delta)bn-1$ random bits from the $\ell$-th block $x_{I_{\ell}}$ and outputs the first $m_{bn}-d_{bn}=\delta b n/2$ bits. The last $d_{bn}=(1+\delta/2)bn-1$ bits, which is exactly the seed used in this block, will be reused as the seed in the next block. Therefore, $\Ext_{\rbs}^{\rmmeq}$ uses $d_{bn}\approx \left(4/\delta+2\right)\log\left(N/\epsilon\right)$ random bits as seed, and outputs $(N/n)\cdot (m_{bn}-d_{bn})\approx \delta bN/2$ bits in total. The total time costed is about $(N/n)\times n\log n=N\log n$ \cite{golub1996matrix}. So when $N$ is sufficiently large (e.g., $N\geq 2^{50}$ corresponding to a 100 Gbps QRNG working for several hours), $\Ext_{\rbs}^\rmmeq$ can achieve one order of magnitude faster than applying the Toeplitz-hashing extractor once to the whole raw data, which costs about $N\log N$ time. In addition, though $\Ext_{\rbs}^{\rmmeq}$ uses more seed than $\Ext_{\rbs}^{\rmmneq}$, it is more hardware-friendly and can achieve much higher extraction speed.


\begin{corollary}
  Suppose the random sources in Algorithms \ref{alg:rbseq} and \ref{alg:rbsneq} are replaced with the $(b,N,\delta,\epsilon_\mathrm{s})$ and $(b,\infty,\delta,\epsilon_\mathrm{s})$-smoothed reverse block sources, respectively. By using the same processing procedures, the output state of the extractor $\Ext_{\rbs}^{\rmmneq}$ satisfies
  \begin{equation}		\frac{1}{2}\|\rho_{r^{(1)}\circ r^{(2)}\circ\cdots\circ r^{(k)}\qstate}-\rho_{U_{\sm_k}}\otimes \rho_{\qstate}\|
  <\frac{2^{-\delta b n_1/4}}{1-2^{-\delta b\Delta/4}}+2\epsilon_s,
  \end{equation}
  and the output state of the extractor $\Ext_{\rbs}^{\rmmeq}$ satisfies
  \begin{equation}
  \frac{1}{2}\|\rho_{r^{(1)}\circ r^{(2)}\circ\cdots\circ r^{(N/n)}\circ\qstate}-\rho_{U_{\eta}}\otimes \rho_{\qstate}\|
  \leq\epsilon+2\epsilon_s.
\end{equation}
\end{corollary}

\begin{proof}
  We prove the smoothed version of Algorithm~\ref{alg:rbsneq}, and the proof for the smoothed version of Algorithm~\ref{alg:rbseq} follows essentially the same procedures. For brevity, denote $\frac{2^{-\delta b n_1/4}}{1-2^{-\delta b\Delta/4}}:=\epsilon$. 
  According to the definition of the smoothed reverse block source, there exists a state $\rho^*_{X\qstate}$ that is $\epsilon_s$-close to the real output state $\rho_{X\qstate}$ such that $\rho^*_{X\qstate}$ determines a $(b,\infty,\delta)$ reverse block source. Using the result in Theorem~\ref{thm:ext_rbs},
\begin{equation}\label{eq:smrbs_converge}
\begin{aligned}
&\frac{1}{2}\left\|\rho_{r^{(1)}\circ r^{(2)}\circ\cdots\circ r^{(k)}\circ s^{(k+1)}\circ X_{I_{k+1:\infty}}\qstate}-\rho_{U_{\sm_k+d_{bn_{k+1}}}}\otimes \rho_{X_{I_{k+1:\infty}}\qstate}\right\| \\
\leq & \frac{1}{2}\left\|\rho_{r^{(1)}\circ r^{(2)}\circ\cdots\circ r^{(k)}\circ s^{(k+1)}\circ X_{I_{k+1:\infty}}\qstate}-\rho^*_{r^{(1)}\circ r^{(2)}\circ\cdots\circ r^{(k)}\circ s^{(k+1)}\circ X_{I_{k+1:\infty}}\qstate}\right\| \\
&+ \frac{1}{2}\left\|\rho^*_{r^{(1)}\circ r^{(2)}\circ\cdots\circ r^{(k)}\circ s^{(k+1)}\circ X_{I_{k+1:\infty}}\qstate}-\rho_{U_{\sm_k+d_{bn_{k+1}}}}\otimes \rho^*_{X_{I_{k+1:\infty}}\qstate}\right\| \\
& + \frac{1}{2}\left\|\rho_{U_{\sm_k+d_{bn_{k+1}}}}\otimes \rho^*_{X_{I_{k+1:\infty}}\qstate} - \rho_{U_{\sm_k+d_{bn_{k+1}}}}\otimes \rho_{X_{I_{k+1:\infty}}\qstate}\right\| \\
\leq &\epsilon_\mathrm{s}+\epsilon+\epsilon_\mathrm{s}=\epsilon +2 \epsilon_\mathrm{s},
\end{aligned}
\end{equation}
where the first inequality is due to the contractivity of trace-preserving quantum operations while the last inequality comes from the fact that the purified distance is an upper bound of the trace distance.
\end{proof}
This result implies that the output for a smoothed reverse block source can also be arbitrarily close to a uniform distributed sequence.

\section{Simulations of the real-time randomness generation rate}\label{sec:simulation}
In this section, we make a simulation estimating the extraction speed of the first extractor $\Ext_{\rbs}^{\rmmeq}$ implemented in the Xilinx Kintex-7 XC7K480T Field Programmable Gate Array (FPGA), a common application in industry. Here, we adopt the expanded Toeplitz-hashing extractor $\Ext^n_{T'}$ as the gadget $\Ext_{\ggg}^n$. We use the raw data from Ref.~\cite{gehring2021homodyne} as the random source, which is a reverse block source with parameters $b=16$ and $\delta=10.74/16\approx 0.67$. We consider a raw data length of $N=2^{51}$ and a total security parameter $\epsilon=2^{-30}$, which means the final output data is $2^{-30}$-close to a uniform distribution. Correspondingly, the block length is $n=\left\lceil 4/(\delta b)\cdot \log\left(N/\epsilon\right)\right\rceil=31$.

According to the discussion after Theorem \ref{thm:ext_rbseq}, $\Ext_{\rbs}^{\rmmeq}$ will use a seed of $(1+\delta/2)bn-1=662$ bits and output about $10.74$ PB random bits. Precisely, $\Ext_{\rbs}^{\rmmeq}$ simply divides the input raw data into blocks with length $31$, where each block contains $b\times n=16*31=496$ bits. The output length is $m = \delta n/2 = 167$ bits. Then the calculation in each block is a $167\times 496$ Toeplitz matrix multiplied by a $496$-dimension vector over $GF(2)$. Since the addition $a+b$ and multiplication $a\cdot b$ over $GF(2)$ are exactly $a\oplus b$ and $a\wedge b$, respectively, which are both basic logical operations, the matrix multiplication involves $167\times 496=82832$ `$\wedge$' operations and $82832$ `$\oplus$' operations.


The parameters of the Xilinx Kintex-7 XC7K480T FPGA are as follows. The clock rate is set to be $200$ MHz; the number of Look-Up-Tables (LUTs) is $3\times 10^5$; each LUT can perform $5$ basic logical operations simultaneously.  To make full use of the FPGA, we can perform the matrix multiplications of $\lfloor 3\times 10^5 \times 5/(2*82832)\rfloor=9$ blocks in parallel. 
Therefore, the extraction speed of $\Ext^{\rmmeq}_{\rbs}$ is $200\times 10^6\times 9\times 167 \approx 300$ Gbps, which is improved by one order of magnitude compared to the state-of-the-art result with a speed of 18.8 Gbps \cite{bai202118}. As a result, the extraction speed is not a bottleneck for the high-speed QRNGs any more; hence our online extraction is adequate for the post-processing of the fastest known implementation of QRNGs \cite{qi2010high,xu2012ultrafast}.

\section{Conclusion}

In conclusion, we design two novel quantum randomness extractors based on the reverse-block-source property that is satisfied by a large class of trusted-device QRNGs. These results provide theoretical supports to the current real-time block-wise post-processing widely applied in experiments and industry. The first extractor improves the real-time extraction speed while the second one can extract infinite raw data with only a constant seed length. In particular, the first extractor is easy to be implemented in a FPGA. The real-time extraction speed with a common FPGA is high enough for the real-time post-processing of current QRNGs.

For future work, it is interesting to explore other properties beyond the general min-entropy source to improve the post-processing. The improvement may come from boosting the extraction speed or saving the seed length. On the other hand, randomness extraction with an imperfect seed or even without seed is also a practical and promising direction. An interesting open question is whether randomness can still be extracted \emph{online} without the seed for certain non-trivial random sources. Finally, it is interesting to apply novel extractors as the gadget in our framework for fewer seeds or faster real-time extraction speed, such as the extractors proposed in \cite{hayashi2016more} and \cite{huang2022stream}.

\section*{Acknowledgments}
We thank Y. Nie and B. Bai for enlightening discussions, and Salil Vadhan for telling us the details of the extractor which extracts randomness from Santha-Vazirani sources using a seed of constant length. This work was supported in part by the National Natural Science Foundation of China Grants No. 61832003, 61872334, 61801459, 62002229, 1217040781, and the Strategic Priority Research Program of Chinese Academy of Sciences Grant No. XDB28000000. Qian Li was additionally supported by Hetao Shenzhen-Hong Kong Science and Technology Innovation Cooperation Zone Project (No.HZQSWS-KCCYB-2022046).
\onecolumngrid

\appendix

\section{QRNG raw data as a reverse block source}\label{app:rbs}

In this section, we prove that the raw data of QRNGs based on homodyne detection~\cite{gehring2021homodyne} satisfy the reverse block source. The settings are shown in Fig.~\ref{fig:vacuum}, where a homodyne detector is employed to measure the $X$ quadrature of a vacuum state. Since the vacuum state is a Gaussian state, the output is supposed to follow a Gaussian distribution. The homodyne detector is concatenated by an analog-to-digital converter, which transforms the continuous voltage signal into discrete raw data. Due to the finite bandwidth of the photo detector, there will be overlaps between different input signals in temporal domain, which leads to correlations in the raw data.
\begin{figure}[htbp]
	\centering
	\includegraphics[scale=0.7]{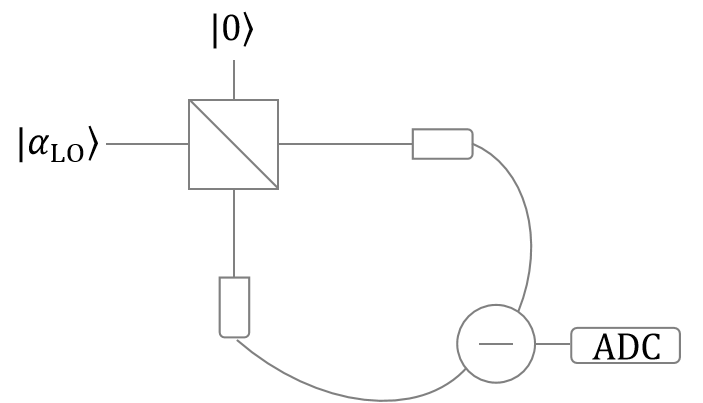}
	\caption{Settings of the QRNG based on homodyne detection. We use $\ket{0}$ and $\ket{\alpha_{LO}}$ to represent the vacuum state and local oscillator, respectively.}
	\label{fig:vacuum}
\end{figure}

We introduce the quantities involved in the QRNG based on homodyne detection. We use $X$ to represent the continuous output of the homodyne detector.
which consists of two components: the quantum quadrature $Q$ and the excess noise $U$, i.e., $X=Q+U$.
Both $X$ and $U$ are assumed to follow stationary Gaussian processes. The output at time $t$ is denoted by $X_t$. For arbitrary $t$, the variance of $X_t$ is given by
\begin{equation}\label{appeq:variance}
    \var(X_t) = \sigma_X^2 + \zeta_t,
\end{equation}
where $\zeta_t$ accounts for the fluctuations of the past outputs $X_{<t}$ and $\sigma_X^2$ accounts for the fluctuations independent of the past, i.e., the variance of $X_t$ conditioned on $X_{<t}$. More specifically, $p(X_1,X_2,...X_t)$ follows a multivariate Gaussian distribution
and 
\begin{equation}\label{appeq:conditional_prob}
p_{X_t}(x_t| X_{<t}=x_{<t})= G(x_t;\mu_t,\sigma_X^2).    
\end{equation}
where $\mu_t$ depends on the past outputs $x_{<t}$ and 
\begin{equation}
G(x;\mu,\nu^2)=\frac{1}{\sqrt{2\pi \nu^2}}e^{-\frac{(x-\mu)^2}{2\nu^2}}.    
\end{equation}
denotes a Gaussian in the variable $x$, with mean $\mu$, and variance $\nu^2$. For the excess noise $U_t$, we can similarly write the probability
density distribution conditioned on past values, i.e.,
\begin{equation}
p_{U_t}(u_t| U_{<t}=u_{<t})= G(u_t;\nu_t,\sigma_U^2).    
\end{equation}
The raw data $X$ is discretized by an $b$-bit ADC. We denote the random variable after coarse-graining as $\bar{X}$, whose probability distribution is given by
\begin{equation}
p_{\bar{X}}(j) =\int _{x\in I_j} p_{X}(x)dx.  
\end{equation}
The intervals of the ADC $I_j$ $(j\in \{1,2,\dotsc, 2^b-1\})$  is determined by
\begin{equation}
    \begin{aligned}
        I_1 & =(-\infty, -R) \\
        I_{2^b} & = (R,\infty) \\
        I_j & = \left(-R+\frac{2R(j-2)}{2^b-2},-R+\frac{2R(j-1)}{2^b-2} \right), \quad j\in\{2,3,\dotsc, 2^b-1\}
    \end{aligned}
\end{equation}
The voltage range $R$ is heuristically chosen. 


Eve's knowledge about different samples is assumed to be independent \cite{gehring2021homodyne}. That is, the joint state of $X_{1}\cdots X_{N}\mathcal{E}$ can be written as
\begin{equation}
 \rho_{X_1\cdots X_{N}\mathcal{E}}=\sum_{x_1\cdots x_{N}}P_{x_1\cdots x_{N}} \ket{x_1\cdots x_{N}}\bra{x_1\cdots x_{N}}\otimes \rho_{\mathcal{E}_1}^{x_1}\otimes \cdots\otimes\rho_{\mathcal{E}_N}^{x_N}.   
\end{equation}
Furthermore, $\mathcal{E}_i$ is assumed to be independent with $x_j$ if $i\neq j$. 

In the following, we will give a lower bound on the min-entropy $H_{\min}(\bar{X}_{k},\bar{X}_{k+1},\cdots,\bar{X}_{i}| \bar{X}_{i+1}=\bar{x}_{i+1},\cdots, \bar{X}_N=\bar{x}_{N},\mathcal{E})$. Lemma \ref{lem:chainrule} will be used.
\begin{lemma}[\cite{gehring2021homodyne}]\label{lem:minent}
Let $Y\mathcal{E}$ be a cq-state where $\rho_{Y\mathcal{E}}=\sum_y P_y \ket{y}\bra{y}\otimes \rho_{\mathcal{E}}^y$. Then
\begin{equation}
 H_{\min}(Y|\mathcal{E})_{\rho}=\sup_{\gamma}\left[-\log_2\sup_{y}\left(P_y\cdot\|\gamma^{-1/2}\rho_{\mathcal{E}}^x\gamma^{-1/2}\|_{\infty}\right)\right],   
\end{equation}
where $\|\cdot\|_{\infty}$ denotes the operator norm (equal to the value of the maximum eigenvalue), and the first supremum is over a density operator $\gamma$ on $\mathcal{E}$.
\end{lemma}
\begin{lemma}\label{lem:chainrule}
Let $(X_1,X_2,\mathcal{E}_1,\mathcal{E}_2)$ be a ccqq-state where 
\begin{equation}
\rho_{X_1X_2\mathcal{E}_1\mathcal{E}_2}=\sum_{x_1,x_2}P_{x_1,x_2}\ket{x_1,x_2}\bra{x_1x_2}\otimes \rho_{\mathcal{E}_1}^{x_1}\otimes \rho_{\mathcal{E}_2}^{x_2}.   
\end{equation}
If the following holds:
\begin{itemize}
\item $H_{\min}(X_1 |\mathcal{E}_1)\geq \delta_1$, and
\item there is a universal density operator $\gamma_2$ on $\mathcal{E}_2$ such that for any possible value $x_1$ of $X_1$,
\begin{equation}
\sup_{x_2}\left(P_{x_2| x_1}\cdot\|\gamma_2^{-1/2}\rho_{\mathcal{E}_2}^{x_2}\gamma_2^{-1/2}\|_{\infty}\right)\leq 2^{-\delta_2}.    
\end{equation}
That is, a universal $\gamma_2$ makes $H_{\min}(X_2| X_1=x_1,\mathcal{E}_2)\geq \delta_2$ for any $x_1$.
\end{itemize}
Then we have 
\begin{equation}
   H_{\min}(X_1,X_2|\mathcal{E}_1,\mathcal{E}_2)\geq \delta_1+\delta_2. 
\end{equation}
\end{lemma}
\begin{proof}
Let $\gamma_1$ denote a density operator on $\mathcal{E}_1$ to make $H_{\min}(X_1| \mathcal{E}_1)\geq \delta_1$, i.e., 
\begin{equation}
\sup_{x_1}\left(P_{x_1}\cdot\|\gamma_1^{-1/2}\rho_{\mathcal{E}_1}^{x_1}\gamma_1^{-1/2}\|_{\infty}\right)\leq 2^{-\delta_1}.  
\end{equation}
Define $\gamma=\gamma_1\otimes \gamma_2$, then we have
\begin{equation}
   \begin{aligned}
&\sup_{x_1,x_2}\left(P_{x_1,x_2}\cdot\|\gamma^{-1/2}\left(\rho_{\mathcal{E}_1}^{x_1}\otimes\rho_{\mathcal{E}_2}^{x_2}\right)\gamma^{-1/2}\|_{\infty}\right)\\
= &\sup_{x_1,x_2}\left(P_{x_1}\|\gamma_1^{-1/2}\rho_{\mathcal{E}_1}^{x_1}\gamma_1^{-1/2}\|_{\infty}\cdot P_{x_2| x_1}\|\gamma_2^{-1/2}\rho_{\mathcal{E}_2}^{x_2}\gamma_2^{-1/2}\|_{\infty}\right)\\
\leq &\sup_{x_1,x_2}\left(P_{x_1}\|\gamma_1^{-1/2}\rho_{\mathcal{E}_1}^{x_1}\gamma_1^{-1/2}\|_{\infty}\right)\cdot \sup_{x_1,x_2} \left(P_{x_2| x_1}\|\gamma_2^{-1/2}\rho_{\mathcal{E}_2}^{x_2}\gamma_2^{-1/2}\|_{\infty}\right)\\
\leq & 2^{-\delta_1-\delta_2}.  
\end{aligned} 
\end{equation}
Then we get the conclusion immediately by Lemma \ref{lem:minent}.
\end{proof}
From the argument in \cite{gehring2021homodyne}, we have that for any $1\leq t\leq N$, there is a universal density operator $\gamma_t$ on $\mathcal{E}_t$ such that for any $\bar{x}_{>t}$,
\begin{equation}\label{eq:base}
\sup_{\bar{x}_t}\left(P_{\bar{x}_t| \bar{x}_{>t}}\cdot\|\gamma_t^{-1/2}\rho_{\mathcal{E}_t}^{\bar{x}_t}\gamma_t^{-1/2}\|_{\infty}\right) = \sup_{\bar{x}_t}\left(P_{\bar{x}_t| \bar{x}_{<t}}\cdot\|\gamma_t^{-1/2}\rho_{\mathcal{E}_t}^{\bar{x}_t}\gamma_t^{-1/2}\|_{\infty}\right)\leq 2^{-\delta^\ast},
\end{equation}
where the first equation is due to the time-reversal symmetry of the temporal signals, the quantities $\delta^*$ and $n$ are given by
\begin{equation}\label{appeq:entropy_lower_bound1}
\begin{aligned}
\delta^\ast& =-\log \left[(\sqrt{n}+\sqrt{n+1})^2\cdot\mathrm{erf}\left(\frac{1}{(2^b-2)g_{\ast}}\right) \right] \\
n& =\frac{1}{2}\frac{\var(X_t)}{\sigma_{X}^2-\sigma_U^2}-\frac{1}{2}
\end{aligned}
\end{equation}
and $g_{\ast}$ is implicitly defined by the equation 
\begin{equation}\label{appeq:entropy_lower_bound2}
\mathrm{erf}\left(\frac{1}{(2^b-2)g_{\ast}}\right)=\frac{1}{2}\mathrm{erf}\left(\frac{R}{g_{\ast}}\right).   
\end{equation}
From Eq.~\eqref{appeq:entropy_lower_bound1}, we can see that the upperbound given by Eq.~\eqref{eq:base} is independent of the mean of $P_{x_t| x_{<t}}$, 
i.e., $\mu_t$ in Eq.~\eqref{appeq:conditional_prob}, but depends on $\sigma_X^2$, $\sigma_U^2$, and $\var(X_t)$, which are identical in each round. Therefore for an arbitrary $t$,
\begin{equation}
H_{\min}(\bar{X}_t| \bar{X}_{>t}=\bar{x}_{>t},\mathcal{E}_t)\geq \delta^\ast.
\end{equation}
Then by apply Lemma \ref{lem:chainrule} repeatedly, we have that for any $\bar{x}_{i+1},\cdots,\bar{x}_{N}$
\begin{equation}
H_{\min}(\bar{X}_{k},\bar{X}_{k+1},\cdots,\bar{X}_{i}| \bar{X}_{i+1}=\bar{x}_{i+1},\cdots, \bar{X}_N=\bar{x}_{N},\mathcal{E})\geq (i-k+1)\delta^\ast.    
\end{equation}

\section{Initial seed is indispensable for reverse block sources}\label{app:seed_is_necessary}
The following theorem claims that there does not exist any deterministic extractor that can extract even one bit of almost uniformly-distributed random number from every reverse block source. The proof is essentially the same as that for Santha-Vazirani sources \cite{sv1,sv2}.
\begin{theorem}\label{lem:no-deterministic}
For all $b,N\in\mathbb{N}^+$ and $0<\delta<1$, and any function  $\Ext:\{0,1\}^{bN}\rightarrow\{0,1\}$, there exists a $(b,N,\delta)$-reverse block source $X$ such that either $\Pr[\Ext(X) = 1]\geq 2^{-\delta}$ or $\Pr[\Ext(X)=1]\leq 1-2^{-\delta}$.
\end{theorem}
\begin{proof}
	Because $|\Ext^{-1}(0)|+|\Ext^{-1}(1)|=2^{bN}$, either $|\Ext^{-1}(0)|\geq 2^{bN-1}$ or $|\Ext^{-1}(1)|\geq 2^{bN-1}$. Without loss of generality, let us assume that $|\Ext^{-1}(1)|\geq 2^{bN-1}$. Pick an arbitrary subset $S$ of $\Ext^{-1}(1)$ with $|S|=2^{bN-1}$. Consider the source $X$ that is uniformly distributed on $S$ with probability $2^{-\delta}$ and is uniformly distributed on $\{0,1\}^{bN}\setminus S$ with probability $1-2^{-\delta}$. It is easy to check that $\Pr[\Ext(X)=1]\geq 2^{-\delta}>1/2$.
	
What remains is to show that $X$ is a $(b,N,\delta)$-reverse block source.  In this proof, we use $X_i$ to stand for the $i$-th bit of the string $X$, and use $X(i)$ to stand for the $i$-th block of $X$, that contains $b$ bits $X_{(i-1)b+1},X_{(i-1)b+2},\cdots,X_{ib}$. 
	
The crucial observation is that $\Pr[X=x]/\Pr[X=x']\leq (2^{-\delta})/(1-2^{-\delta})$ for any $x,x'\in\{0,1\}^{bN}$. In particular, for every $i\in [bN]$, and every $x_1,\cdots,x_{i-1},x_{i+1},\cdots,x_{bN}$, we have
\begin{equation}\label{eq:app_b1}
\begin{aligned}
\frac{1-2^{-\delta}}{2^{-\delta}}\leq \frac{\Pr[X=x_1\cdots x_{i-1}0x_{i+1}\cdots x_{bN}]}{\Pr[X=x_1\cdots x_{i-1}1x_{i+1}\cdots x_{bN}]} \leq \frac{2^{-\delta}}{1-2^{-\delta}}.
\end{aligned}
\end{equation}
Combining with the facts that 
\begin{equation}
\begin{aligned}
\Pr[X_i=0,X_{i+1}=x_{i+1},\cdots,X_{bN}=x_{bN}]=\sum_{X_1\cdots X_{i-1}\in\{0,1\}^{i-1}}\Pr[X=x_1\cdots x_{i-1}0x_{i+1}\cdots x_{bN}]
\end{aligned}
\end{equation}
and 
\begin{equation}
\begin{aligned}
\Pr[X_i=1,X_{i+1}=x_{i+1},\cdots,X_{bN}=x_{bN}]=\sum_{X_1\cdots X_{i-1}\in\{0,1\}^{i-1}}\Pr[X=x_1\cdots x_{i-1}1x_{i+1}\cdots x_{bN}]
\end{aligned}
\end{equation}
and inequality Eq.~\eqref{eq:app_b1}, we have 
\begin{equation}
\begin{aligned}
\frac{1-2^{-\delta}}{2^{-\delta}}\leq \frac{\Pr[X_i=0,X_{i+1}=x_{i+1},\cdots,X_{bN}=x_{bN}]}{\Pr[X_i=1,X_{i+1}=x_{i+1},\cdots,X_{bN}=x_{bN}]} \leq \frac{2^{-\delta}}{1-2^{-\delta}}.
\end{aligned}
\end{equation}
This is equivalent to saying that for any $x_i,x_{i+1},\cdots,x_{bN}$, 
\begin{equation}\label{eq:app_b2}
\begin{aligned}
\Pr[X_i=x_i| X_{i+1}=x_{i+1},\cdots,X_{bN}=x_{bN}]\leq 2^{-\delta}.
\end{aligned}
\end{equation}

Finally, we have for any $k_1,k_2\in[N]$ where $k_1\leq k_2$, and any $x(k_1),x(k_1+1),\cdots,x(N)$,
\begin{equation}
\begin{aligned}
&\Pr[X(k_1)=x(k_1),\cdots,X(k_2)=x(k_2)| X(k_2+1)=x(k_2+1),\cdots,X(N)=x(N)]\\
=&\Pr[X_{(k_1-1)b+1}=x_{(k_1-1)b+1},\cdots,X_{k_2b}=x_{k_2b}| X(k_2+1)=x(k_2+1),\cdots,X(N)=x(N)]\\
=&\Pi_{i=(k_1-1)b+1}^{k_2b} \Pr[X_{i}=x_{i}| X_{i+1}=x_{i+1},\cdots,X_{bN}=x_{bN}]\\
\leq & \Pi_{i=(k_1-1)b+1}^{k_2b} 2^{-\delta}=2^{-(k_2-k_1+1)\delta b},
\end{aligned}
\end{equation}
where the inequality is by inequality Eq.~\eqref{eq:app_b2}.
Now we can conclude that $X$ is a $(b,N,\delta)$-reverse block source by Definition \ref{def:rbs}.

\end{proof}

\bibliographystyle{apsrev4-1}

\bibliography{ref}


\end{document}